\documentclass[12pt]{article}

\usepackage{amsmath}
\usepackage{amssymb}
\usepackage{amsfonts}
\usepackage{latexsym}
\usepackage{color}
\usepackage[normalem]{ulem}
\usepackage{cancel}

\catcode `\@=11 \@addtoreset{equation}{section}

\catcode `\@=12

\newcommand{\be}{\begin{equation}}
\newcommand{\en}{\end{equation}}
\newcommand{\bea}{\begin{eqnarray}}
\newcommand{\ena}{\end{eqnarray}}
\newcommand{\beano}{\begin{eqnarray*}}
\newcommand{\enano}{\end{eqnarray*}}
\newcommand{\bee}{\begin{enumerate}}
\newcommand{\ene}{\end{enumerate}}

\newcommand{\mb}{\mathbb}

\newcommand{\mc}{\mathcal}

\newcommand{\D}{{\mc D}}
\newcommand{\LL}{\mc L}

\newcommand{\Sc}{{\cal S}}
\newcommand{\E}{{\cal E}}
\newcommand{\F}{{\cal F}}
\newcommand{\G}{{\cal G}}

\newcommand{\Lc}{{\cal L}}

\newcommand{\1}{1 \!\! 1}

\newcommand{\V}{{\cal V}}
\newcommand{\X}{{\cal X}}
\newcommand{\ip}[2]{\langle {#1},{#2}\rangle}

\newcommand{\Hil}{\mc H}

\newtheorem{thm}{Theorem}

\newtheorem{prop}[thm]{Proposition}
\newtheorem{defn}[thm]{Definition}

\newtheorem{rem}[thm]{Remark}
\newenvironment{proof}{\noindent {\bf Proof. }}{\hfill$\square$ \vspace{3mm}\endtrivlist}
\newcommand{\balpha}{{\mbox{\boldmath${\alpha}$}}}
\newcommand{\bgamma}{{\mbox{\boldmath${\gamma}$}}}

\newcommand{\bbeta}{{\mbox{\boldmath${\beta}$}}}
\newcommand{\balphabar}{{\overline{\mbox{\boldmath${\alpha}$}}}}

\catcode `\@=11 \@addtoreset{equation}{section}
\catcode `\@=12

\begin{document}

\thispagestyle{empty}

\vspace*{2cm}

\begin{center}
{\Large \bf {Hamiltonians defined by biorthogonal sets } }  \vspace{2cm}\\

{\large Fabio Bagarello}\\
  Dipartimento di Energia, Ingegneria dell'Informazione e Modelli Matematici,\\
  Universit\`a di Palermo, and INFN, Torino\\
e-mail: fabio.bagarello@unipa.it\\
Tel: +390912389722; Fax: +39091427258\\
home page: www.unipa.it/fabio.bagarello

\vspace{.4cm}

 {\large Giorgia Bellomonte}\\
  Dipartimento di Matematica e Informatica,\\ Universit\`a di Palermo, I-90123 Palermo, Italy\\e-mail: giorgia.bellomonte@unipa.it

\end{center}

\vspace*{2cm}

\begin{abstract}
\noindent In some recent papers, the studies on biorthogonal Riesz
bases has found a renewed motivation because of their connection
with pseudo-hermitian Quantum Mechanics, which deals with physical
systems described by Hamiltonians which are not self-adjoint but
still may have real point spectra. Also, their eigenvectors may form
Riesz, not necessarily orthonormal, bases  for the Hilbert space in
which the model is defined. Those Riesz bases allow a decomposition
of the Hamiltonian, as already discussed is some previous papers. However, in many physical models, one has to
deal not with  o.n. bases or with Riesz bases,  but just with
biorthogonal sets. Here, we consider the more general concept of
$\mathcal{G}$-quasi basis and we show a series of conditions under
which a definition of non self-adjoint Hamiltonian with purely point
real spectra is still possible.

\end{abstract}

\vspace{2cm}

%{\bf PACS Numbers}:  .......

\vfill

%\pagenumbering{roman}

\newpage

\section{Introduction and preliminaries}
For several years  physicists have devoted their studies to those
systems which were described by self-adjoint Hamiltonians. This
choice have been led by the fact that the eigenvalues of a
Hamiltonian describing a physical system represent the energy of
that system hence they must be real to have a physical meaning and self-adjoint Hamiltonians
have real eigenvalues. This is important to ensure that the dynamics of the system is unitary, so that the probability described by the wave-function is preserved during the time evolution. In recent
years many physicists (as Bender and his collaborators) first, and
mathematicians after, started to consider with more and more
interest non self-adjoint Hamiltonians with real spectra because
they described some physical system. The beginning of the story goes probably back to the paper \cite{ben1}, in which the eigenstates of the manifestly non self-adjoint Hamiltonian $H=p^2+ix^3$ were deduced and found to be real. Here $x$ and $p$ are the position and momentum operators, satisfying the Weyl algebra. A very recent book on this and related topics is \cite{bagbook_thebook}.

The key objects in that analysis  were the
so-called PT-symmetric Hamiltonians with real point spectra.  A
PT-symmetric Hamiltonian is an operator such that
$$ PTH(PT)^{-1} = PTHPT = H,$$ where $P$ and $T$ are respectively
the operators of parity and time-reversal transformations, usually defined\footnote{In the physical literature the definition of $P$ and $T$ really depends on the particular model under consideration, and can change quite a bit, from model to model.}
according to $P x P = -x$, $P p P = T p T = -p$, $T i\1 T = -i\1$,
where $x, p, \1$ are respectively the position, momentum, and
identity operators acting on the Hilbert space
$\Hil=L^2(\mathbb{R})$ and $i$ is the square root of $-1$.   Later was understood that
PT-symmetry can be replaced by more general requirements, still
getting the same conclusion: at the beginning of this century Mostafazadeh
introduced the concept of pseudo-hermitian (also called
pseudo-symmetric or quasi-Hermitian) operators as those operators satisfying, in some sense, an intertwining relation of the form $AG=GA^\dagger$. Here $A$ is the pseudo-hermitian operator, while $G$ is a certain positive operator. Notice that, quite often, both $A$ and $G$ are unbounded, so that the equality $AG=GA^\dagger$ is only formal. However, this equality was recently made rigorous in \cite{AnTr2014}, where $A$ and $G$  are defined  as those operators
$A$, with dense domain $D(A)$ for which there exists a positive
operator $G$, with dense domain $D(G)$ in Hilbert space $\Hil$ such
that $D(A)\subset D(G)$ and
\begin{equation} \label{eq_quasihermitian}
\ip{A\xi}{G\eta}= \ip{G\xi}{A\eta}, \quad \xi, \eta \in
D(A).\end{equation} Some recent results show that, even if reality
of the eigenvalues of a certain Hamiltonian is ensured, the basis
property of its eigenstates is, in many cases, lost. {In other
words, if a non self-adjoint operator $H$ in the Hilbert space
$\Hil$ has purely punctual real spectrum, and if $H$ is PT-symmetric or,
more generally, pseudo-symmetric, it is not necessarily true that
the set of eigenstates of $H$, $\F_{\varphi}=\{\varphi_n\in\Hil\}$
and of its adjoint $H^\dagger$, $\F_{\Psi}=\{\Psi_n\in\Hil\}$, are
biorthogonal bases for $\Hil$}. Indeed, this feature was already
discussed in two papers by Davies, \cite{dav1,dav2},
 and then in several papers by one of us (FB), see \cite{bagbook} for a recent review, in \cite{petr},
 and in other recent papers. { The importance of the basis property is obvious. For our present purposes, it mainly lays in the fact that the operators which has those bases as sets of eigenvectors can
 be decomposed in terms of those sets.}  Still, quite often, the sets of eigenvectors (although are not a basis of the Hilbert space)  are  {\em complete} in $\Hil$, meaning that the only vector which is orthogonal to
 all the $\varphi_n$'s or to all the $\Psi_n$'s is the zero vector.

Recent literature has dealt with, in a certain sense, the inverse problem than that before. In \cite{fg,bit,cg}, the problem of considering some particular
biorthogonal sets of vectors to define non self-adjoint operators,
has been considered,
leading to a number of interesting results as factorizability of Hamiltonians
by some kind of lowering and raising operators. In particular, in \cite{bit} Riesz bases of a given Hilbert space have been used, while in \cite{cg} and, later, in \cite{fg},
the interest was focused again on a generalization of the notion of Riesz basis, living in a rigged Hilbert space.
In \cite{bag1}, the notion of $\G$-quasi bases  is considered.
 As we will see in the following section, two biorthogonal sets  $\F_\varphi=\{\varphi_n\in\Hil,\,n\geq0\}$ and
$\F_\Psi=\{\Psi_n\in\Hil,\,n\geq0\}$, satisfying
$\left<\varphi_n,\Psi_m\right>=\delta_{n,m}$, are called {\em
$\G$-quasi bases} if, for all $f, g\in \G$, the following holds: \be
\left<f,g\right>=\sum_{n\geq0}\left<f,\varphi_n\right>\left<\Psi_n,g\right>=\sum_{n\geq0}\left<f,\Psi_n\right>\left<\varphi_n,g\right>,
\label{01}\en  being $\G$ a dense subset of the Hilbert space $\Hil$. The physical relevance of these families of vectors is that in some physical system driven by some Hamiltonian $H$, see \cite{bagbook}, the set of eigenstates of $H$ and $H^\dagger$ turn out to be $\G$-quasi bases, even if they are and bases.

The problem of generalizing those results to biorthogonal $\G$-quasi bases raises naturally, and in fact this is the main content of this paper, where we will discuss how some of the general ideas introduced in \cite{bit} (we will work on a given Hilbert space $\Hil$, leaving a possible extension to rigged Hilbert spaces to a future analysis) still work even if Riesz bases are replaced by $\G$-quasi bases. This generalization requires us to go into two parallel directions: first of all, since $\G$-quasi bases have only be introduced recently, \cite{bag1}, and since several unusual features are related to these sets of vectors, we will describe in some details three examples of them. Secondly, physical applications of $\G$-quasi bases show that unbounded operators (and intertwining operators in particular, see Section IV) become relevant,
 in this context.   So we will take care of this aspect and we will also discuss in some details what happens when one deals with $\G$ quasi-bases and, therefore, when a resolution of the identity can only be introduced in a weak form, as in (\ref{01}).

The paper is organized as follows: in Section II we introduce some useful definitions on bases and biorthogonal sets and we shortly review the results in \cite{bit} on Hamiltonians defined by Riesz bases. In Section III we give examples of $\G$-quasi bases, and we start considering a few physical consequences, which are analyzed further in Section IV. Section V contains our conclusions.

\section{Some useful definitions and results}\label{sect2}
{Let us begin by recalling some well known definitions. Let  $\Hil$
be a Hilbert space with scalar product $\ip{\cdot}{\cdot}$ linear in
the second entry. We recall the following definitions.
\begin{defn}
{The sequence $\X=\{\xi_n\in\Hil, n\geq0\}$ is said to be complete
if $\ip{\varphi}{\xi_n}=0$ for every $n\in\mathbb{N}$, with $\varphi\in\Hil$, then $\varphi=0$.}

\end{defn}

This means that the set $\X$ is if the only vector $\varphi\in\Hil$ which is orthogonal to all the
$\xi_n$'s is necessarily the zero vector. Sometimes in the literature, rather than {\em complete} the word {\em total} is adopted.

In our analysis the following, slightly modified, version of
completeness will  be also used.
{\begin{defn}
Let
$\F=\{f_n\in\Hil,\,n\geq0\}$ be a set and $\V\subseteq\Hil$ a subspace,  we
will say that $\F$ is complete  in $\V$ if, taken $\varphi\in \V$
such that $\left<\varphi,f_n\right>=0$ for all $n\geq0$, then
$\varphi=0$. \end{defn}
In particular, if $\V=\Hil$, then we will simply say
that $\F$ is complete.}
{
\begin{defn}
A set $\E=\{e_n\in\Hil,\,n\geq0\}$ is said to be a (Schauder)
\index{basis}basis for $\Hil$ if for every $f\in\Hil$ there exists  a {\em unique} sequence $\{c_n(f)\}$ of complex numbers (depending on the vector $f$) such that
\begin{equation}\label{decomp} f = \sum_{n=0}^\infty c_n(f) e_n. \end{equation}
\end{defn}

\begin{defn}A (Schauder) basis for $\Hil$, $\E=\{e_n\in\Hil,\,n\geq0\}$,  is said
to be an orthonormal (o.n.) basis for $\Hil$ if
$\left<e_n,e_m\right>=\delta_{n,m}$, for every $n,m\geq 0$.
\end{defn}If  $\E=\{e_n\in\Hil,\,n\geq0\}$  is  an orthonormal (o.n.) basis for $\Hil$, then the coefficient $c_n(f)$ in (\ref{decomp}) can be written as $c_n(f)=\left<e_n,f\right>$, for every $n\geq0$ and for every $f\in\Hil$.}

\begin{defn}\label{defRB} A set $\F=\{f_n\in\Hil,\,n\geq0\}$ is said to be a Riesz basis for
$\Hil$ if there exists a bounded operator $T$ on $\Hil$, with
bounded inverse, and an orthonormal basis
$\E=\{e_n\in\Hil,\,n\geq0\}$ of $\Hil$, such that $f_n=Te_n$, for
all $n\geq0$.
\end{defn}

\begin{rem}It is clear that if $\F$ is a Riesz basis, then in
general $\left<f_n,f_m\right>\neq\delta_{n,m}$, for every $n,m\geq 0$. Moreover, if $T$ and $\E$ are as in Definition \ref{defRB}, the set of vectors $\LL=\{l_n=(T^{-1})^\dagger
e_n,\,n\geq0\}$ is a Riesz basis  for $\Hil$ as well {(called dual basis)},
and $\left<f_n,l_m\right>=\delta_{n,m}$ for all $n,m\geq0$, { i.e.} $\F$
and $\LL$ are {\em biorthogonal}. Here the symbol $\dagger$ indicates
the adjoint with respect to the {\em natural} scalar
product\footnote{The word {\em natural} is used since, quite often,
in Physics literature on $\mathcal{PT}$-quantum mechanics other
scalar products are also introduced.} $\left<.,.\right>$ in $\Hil$.
In these hypotheses, any vector $h\in\Hil$ can be expanded as
follows: \be
h=\sum_{n=0}^\infty\left<f_n,h\right>\,l_n=\sum_{n=0}^\infty\left<l_n,h\right>\,f_n.
\label{11}\en  Here and in the remainder of the paper the convergence of the various series is always assumed to be  unconditional.
Of course, any o.n. basis is a Riesz basis, with $T=\1$ the identity
operator on $\Hil$. In this case the three sets above just collapse:
$\E=\F=\LL$.
\end{rem}

The same expansion as in \eqref{11} holds when $\F$ and
$\LL$ are { biorthogonal} bases, but not necessarily of the Riesz kind.
Notice that each basis $\F$ in $\Hil$ possesses a unique
biorthogonal set $\V$ which is also a basis: hence the expansion in
\eqref{11} is unique, \cite[Theorem 3.3.2]{chri}. If
$\F=\{f_n\in\Hil,\,n\geq0\}$ is a basis (in any of the senses
considered so far), then $\F$ is complete, while the converse is not true, in
general.
In fact, while for o.n. sets completeness is equivalent to the basis property, for non o.n. sets this is false,  see e.g. \cite[Section 3.2.1]{bagbook}.

\begin{rem}

It should be observed that equation (\ref{11}), as many others results in the rest of the paper, make sense, in principle, also in the context of frame theory, see for instance \cite{chri,dau,dau2,heil}. However, we will not  discuss the relation between $\G$-quasi bases and frames here.
In fact, we are more interested to those sets having no redundancy, in contrast with what happens in frame theory, since, while this aspect is surely important for
signal analysis, it is not so important for quantum mechanics, which is our main
interest here. We will comment something more on this aspect along the paper.

\end{rem}

For convenience of the reader we also recall the following definition.

\begin{defn} Let
$\F=\{f_n\in\Hil,\,n\geq0\}$ and $\LL=\{l_n\in\Hil,\,n\geq0\}$ be two
biorthogonal sets. { The projection operator $P_k$ is defined as
$P_kf=\left<f_k,f\right>l_k$}.\end{defn}
 \begin{rem}\label{rem proj}If $\F$ is a basis, both $P_k$ and $\sum_{k=1}^NP_k$ are uniformly bounded. Viceversa, if $\F$ is complete, and if $\sum_{k=1}^NP_k$
  is uniformly bounded, then $\F$ is a basis, \cite[Lemma 3.3.3]{dav}. This fact will be used later on. In particular, the norm of $P_k$ and those of $f_k$ and $l_k$ are related by the following relation:
$ 1\leq\|P_k\|=\|f_k\|\,\|l_k\|, $ for all $k$.\end{rem}

Another absolutely non trivial difference between o.n. and not o.n.
sets has to do with the possibility of extending the basis property
from a dense subset of $\Hil$ to the whole Hilbert space. Let
$\E=\{e_n\in\Hil,\, n\geq0\}$ be an o.n. set and let $\V$ be a dense
subspace of $\Hil$. Suppose  that each vector $f\in\V$ can be
written as follows:
 $f=\sum_n\left<e_n,f\right>\,e_n$, then $\E$ is an o.n. basis for $\V$ and, furthermore, all the vectors in $\Hil$, and not only those in $\V$, admit a similar expansion: $\hat f=\sum_n\left<e_n,\hat f\right>\,e_n$, $\forall\,\hat f\in\Hil$. Hence $\E$ is also an o.n. basis for $\Hil$ \cite[Theorem 3.4.7]{hansen}.
  Let us now replace the o.n. set $\E$ with a second, no longer o.n.,
set $\X=\{x_n\in\Hil,\, n\geq0\}$, and let us again assume that
every $f\in\V$ can be written, in an unique way, as
$f=\sum_n\,c_n(f)\,x_n$, for certain coefficients $c_n(f)$ depending
on $f$. Then, explicit counterexamples show that a similar expansion
does not hold in general for all vectors in $\Hil$, see \cite[Section 3.2.1]{bagbook}. This will be evident also in the examples
considered later on in this paper.
Then we see once again that loosing orthonormality\footnote{please recall that we are not considering any redundancy here, as we have already stressed above.}  produces new and
often undesired mathematical consequences.
Moreover, the
lack of orthonormality has physical consequences as well,
\cite{bagbook,dav1,dav2,petr}. In fact, when a non self-adjoint
Hamiltonian $H$  with only punctual real spectrum is considered, the set
of its eigenstates, $\F_{\varphi}=\{\varphi_n\in\Hil\}$, is not
an o.n. one, in general. Therefore, we are forced to deal with problems
similar to those discussed above. In particular,   even if
$\F_{\varphi}$  were complete, it would not  be necessarily  a
basis for $\Hil$. Also, even if each vector in a dense subspace of $\Hil$ can be linearly expanded in terms of the $\varphi_n$'s, a similar expansion may not be true in all of $\Hil$. For this reason, in connection with several recent applications, see \cite{bagbook, bag1, bag0} and references therein,  the notion of
$\G$-quasi basis was introduced and used heavily { in the analysis of the eigenvectors and the eigenvalues of certain non self-adjoint Hamiltonians}.

\begin{defn}
Let $\G$ be a dense subspace of Hilbert space $\Hil$. Two
biorthogonal sets $\F_\varphi=\{\varphi_n\in\Hil,\,n\geq0\}$ and
$\F_\Psi=\{\Psi_n\in\Hil,\,n\geq0\}$
($\left<\varphi_n,\Psi_m\right>=\delta_{n,m}$, for every $n,m\geq 0$), are called {\em
$\G$-quasi bases} if, for all $f, g\in \G$, the following holds: \be
\left<f,g\right>=\sum_{n\geq0}\left<f,\varphi_n\right>\left<\Psi_n,g\right>=\sum_{n\geq0}\left<f,\Psi_n\right>\left<\varphi_n,g\right>.
\label{12} \en When $\G=\Hil$, i.e. when (\ref{12}) holds for all
$f,g\in\Hil$, then $\F_\varphi$ and $\F_\Psi$ are simpler called
{\em quasi bases}.\end{defn}From formula \eqref{12}, it follows
immediately that, if $f\in \G$ is orthogonal to all the $\Psi_n$'s
(or to all the $\varphi_n$'s), then $f$ is necessarily zero and, as
a consequence, $\F_\Psi$ (or $\F_\varphi$)  is complete in $\G$.
Indeed, using (\ref{12}) with $g=f\in\G$, {if $\left<\Psi_n,f\right>=0$ (or $\left<f,\varphi_n\right>=0$)
for all $n$, we find
$\|f\|^2=\sum_{n\geq0}\left<f,\varphi_n\right>\left<\Psi_n,f\right>=0$}. Therefore  $\|f\|=0$, so that  $f=0$.}

When  $\F_\varphi$ and $\F_\Psi$ are  quasi bases, then they  are
complete. We refer to \cite{bagbook} for more details on $\G$-quasi
bases. Here we just want to observe that equation (\ref{12}) can be
seen as a weak version of a resolution of the identity,
which turns out to be very important, if not essential, in several
physical applications, as, just to cite one, in quantization
problems, see Part 2 of \cite{gazbook}, i.e. when one wants to replace a classical system with its quantum counterpart\footnote{For instance, a classical harmonic oscillator with energy (in suitable units) $E=\frac{1}{2}(p^2+x^2)$ can be {\em quantized} replacing the time-depending functions $x(t)$ and $p(t)$ with the operators $\hat X:=\sum_{n,m=1}^\infty\left(\int_{\Bbb R} \overline{e_n(x)}\,xe_m(x)\,dx\right) P_{n,m}$ and $\hat P:=\sum_{n,m=1}^\infty\left(\int_{\Bbb R} \overline{e_n(x)}\,\left(-i\frac{d}{dx}\right)e_m(x)\,dx\right) P_{n,m}$. Here $e_n(x)=\frac{1}{\sqrt{2^n\,n!\sqrt{\pi}}}\,H_n(x)\,e^{-\frac{x^2}{2}}$, $H_n(x)$ being the $n$-th Hermite polynomial, and $P_{n,m}$ is the operator defined by $(P_{n,m}f)(x)=\left<e_m,f\right>e_n(x)$, for all $f(x)\in \Lc^2(\Bbb R)$. Of course, one should pay attention to the convergence of the series appearing in the definition of $\hat X$ and $\hat P$. We refer to \cite{gazbook} for more details.}
\vspace{2mm}

\begin{rem} It could happen  that it is easier to check that $f$ is orthogonal to, say, all the $\varphi_{2n}$'s and to all the $\Psi_{2n+1}$'s. Then, again, $f=0$, for similar reasons. Of course, this implies in turns that $f$ is also orthogonal to all the $\varphi_{2n+1}$'s and to all the $\Psi_{2n}$'s.
\end{rem}

\subsection{Working with Riesz bases}
\label{sectRB}

We now briefly review what is the role of Riesz bases in our scheme. Let $\F_\phi=\{\phi_n\in\Hil,\,n\geq0\}$ be a Riesz
basis in the Hilbert space $\Hil$ and let  $\F_\psi=\{\psi_n\in\Hil,\,n\geq0\}$ be its dual biorthogonal Riesz basis. In \cite{bit}
  the following operators have been introduced {and their properties have been studied in details}:

$$ \left\{\begin{array}{l} D(H_{\phi, \psi}^\balpha)=\left\{f \in \Hil; \sum_{n=0}^\infty \alpha_n \left<\psi_n,f\right>\phi_n \mbox{ exists in }\Hil \right\}\\
{ } \\
H_{\phi, \psi}^\balpha f= \sum_{n=0}^\infty \alpha_n
\left<\psi_n,f\right>\phi_n, \; f \in D(H_{\phi, \psi}^\balpha)
\end{array}\right.
$$
and
$$ \left\{\begin{array}{l} D(H_{\psi, \phi}^\balpha)=\left\{f \in \Hil; \sum_{n=0}^\infty \alpha_n \left<\phi_n,f\right>\psi_n \mbox{ exists in }\Hil \right\}\\
{ } \\
H_{\psi, \phi}^\balpha f= \sum_{n=0}^\infty \alpha_n
\left<\phi_n,f\right>\psi_n, \; f \in D(H_{ \psi,\phi}^\balpha)
\end{array}\right. .
$$
Here $\balpha=\{\alpha_n\}$ is any sequence of complex numbers. Once again we recall that the convergence of the series is, as everywhere in this paper, the unconditional one.

\vspace*{1mm}

Once we put $\D_\phi:= \mbox{span}\{\phi_n\}$ and $\D_\psi:=
\mbox{span}\{\psi_n\}$,  it is clear that $\D_\phi\subset D(H_{\phi,
\psi}^\balpha)$, $\D_\psi \subset D(H_{\psi, \phi}^\balpha)$ and
that $H_{\phi, \psi}^\balpha \phi_k =\alpha_k \phi_k$, and $H_{\psi,
\phi}^\balpha \psi_k =\alpha_k \psi_k$, $k\geq 0$, so that
$\phi_k$ and $\psi_k$ are eigenstates of $H_{\phi, \psi}^\balpha$
and $H_{\psi, \phi}^\balpha$ respectively, with the same
eigenvalues. Hence, in particular, since $\F_\phi$ and $\F_\psi$ are
bases for $\Hil$, $H_{\phi, \psi}^\balpha$ and $H_{\psi,
\phi}^\balpha$ are densely defined, closed, and $\left(H_{\phi,
\psi}^\balpha\right)^\dagger=H_{\psi, \phi}^\balphabar$, where
$\balphabar=\{\overline{\alpha}_n\}$. Moreover
 $H_{\phi, \psi}^\balpha$ is bounded if and only if $H_{\psi, \phi}^\balpha$ is bounded and this is true if and only if $\balpha$ is a bounded sequence.
In particular $H_{\phi, \psi}^{ \textbf{1}}=H_{\psi, \phi}^{ \textbf{1}}=\1$,
where ${\textbf{1}}$ is the sequence constantly equal to $1$. Moreover, the spectra of $H_{\phi, \psi}^\balpha$ and $H_{\psi, \phi}^\balpha$ are real if and only if each $\alpha_n$ is real.

\begin{rem}

It is worth to notice that the operators $H_{\phi, \psi}^\balpha$ and $H_{\psi, \phi}^\balpha$ introduced above looks quite similar to the so called {\em multipliers}, see \cite{bal,bal2}, which are operators of the form
$$
M_{m,\Phi,\psi}f=\sum_nm_n\left<\psi_n,f\right>\Phi_n,
$$
where $\{\Phi_n\}$ and $\{\Psi_n\}$ are fixed sequences in the Hilbert space, while $\{m_n\}$ is a sequence of scalars. The interest in \cite{bal,bal2} was on mathematical aspects of these multipliers, like, for instance, their unconditional convergence. On the other hand, we are more interested in the role of $\G$-quasi bases in the definition of our (physically-motivated) multipliers.

\end{rem}

\medskip

In a similar way, introducing a second sequence of complex numbers,
$\bbeta:=\{\beta_n\}$, the following operators can  be also defined:

$$ \left\{\begin{array}{l} D(S_\phi^{\bbeta})=\left\{f \in \Hil; \sum_{n=0}^\infty \beta_n \ip{\phi_n}{f}\phi_n \mbox{ exists in }\Hil \right\}\\
{ } \\
S_\phi^{\bbeta} f= \sum_{n=0}^\infty \beta_n \ip{\phi_n}{f}\phi_n,
\; f \in D(S_\phi^{\bbeta})
\end{array}\right.
$$
and
$$ \left\{\begin{array}{l} D(S_\psi^{\bbeta})=\left\{f \in \Hil; \sum_{n=0}^\infty \beta_n \ip{\psi_n}{f}\psi_n \mbox{ exists in }\Hil \right\}\\
{ } \\
S_\psi^{\bbeta} f= \sum_{n=0}^\infty \beta_n \ip{\psi_n}{f}\psi_n,
\; f \in D(S_\psi^{\bbeta}).
\end{array}\right.
$$
It is clear that
\begin{align}\label{2.5} &\D_\psi\subset D(S_\phi^{\bbeta})\quad\mbox{ and }\quad S_\phi^{\bbeta} \psi_k =\beta_k \phi_k, \; k\geq0\,\, ;\\
& \D_\phi\subset D(S_\psi^{\bbeta}) \quad\mbox{ and }\quad
S_\phi^{\bbeta} \phi_k =\beta_k \psi_k, \; k\geq0 .
\label{2.6}
\end{align}
Hence, in particular, $S_\phi^{\bbeta}$ and $S_\psi^{\bbeta}$ are
also densely defined, again due to the fact that $\F_\phi$ and
$\F_\psi$ are bases,  and, see \cite[Proposition 2.2]{bit}, they are
closed and self adjoint if each  $\beta_n\in\mathbb{R}$.
Furthermore, $S_\phi^{\bbeta}$ is bounded if and only if
$S_\psi^{\bbeta}$ is bounded and this is true if and only if
$\bbeta$ is a bounded sequence. Moreover, if $\beta_n=1$ for all
$n\geq0$, then $S_\phi:=S_\phi^{\textbf{1}}$ and
$S_\psi:=S_\psi^{\textbf{1}}$ are bounded positive self-adjoint
operators on $\Hil$ and they are inverses of each other\footnote{$S_\psi$ and $S_\phi$ are usually called {\em frame operators} in the literature of frames, or {\em metric operators} in quantum mechanics, while the operators $S_\phi^{\bbeta}$ and $S_\psi^{\bbeta}$ are particular cases of {\em Riesz bases multipliers}, \cite{bal2}.}, that is
$S_\phi= (S_\psi)^{-1}$. Also, {see \cite[Proposition
2.3]{bit} }$
 S_\psi H_{\phi, \psi}^\balpha= H_{\psi, \phi}^\balpha S_\psi =S_\psi^\balpha$, and
$ S_\phi H_{\psi, \phi}^\balpha= H_{\phi, \psi}^\balpha S_\phi=S_\phi^\balpha$, which are useful intertwining relations. More details on these operators and their domains can be found in \cite{bit}.

\begin{rem}

The relevance of intertwining relations and of the intertwining operators is discussed in detail, for instance, in \cite{bagint,kuru1,kuru2,samani}. They turn out to be very useful in the deduction of the eigenvalues and eigenvectors for certain pairs of Hamiltonians, obeying suitable intertwining relations. When this happens, the Hamiltonians turn out to be isospectral, and their eigenvectors are mapped ones into the others by the intertwining operator itself.

\end{rem}

\vspace{2mm}

Going back to the general case, let $\balpha=\{\alpha_n\}$ be a
sequence of complex numbers. We define the operators ${\sf h}_{\phi,
\psi}^\balpha$ and ${\sf h}_{\psi, \phi}^\balpha$ as follows:
\begin{equation}\label{2add2}\left\{\begin{array}{l}
 {\sf h}_{\phi, \psi}^\balpha =S_\psi^{1/2} H_{\phi, \psi}^\balpha S_\phi^{1/2},\\
{\sf h}_{\psi, \phi}^\balpha = S_\phi^{1/2} H_{\psi, \phi}^\balpha
S_\psi^{1/2}.
\end{array}\right. \end{equation}

 Then, see \cite[Proposition 2.4]{bit},  $D({\sf h}_{\phi,
\psi}^\balpha)=\{ S_\psi^{1/2}f; f \in D(H_{\phi, \psi}^\balpha)\}$,
$D({\sf h}_{\psi, \phi}^\balpha)=\{ S_\phi^{1/2}f; f \in D(H_{\psi,
\phi}^\balpha)\}$ and these are both dense in $\Hil$. Moreover
 $({\sf h}_{\phi, \psi}^\balpha)^* = {\sf h}_{\psi, \phi}^\balphabar$. Finally, and very important, if $\{\alpha_n\}\subset {\mb R}$, then ${\sf h}_{\phi, \psi}^\balpha$ is self-adjoint.

More results can be found in \cite[Section III]{bit}, where ladder
(i.e. raising and lowering) operators  are also introduced in terms
of the Riesz bases $\F_\phi$ and $\F_\psi$. In the rest of the paper
we will discuss, both from a general point of view and considering
particular examples, how much of the above structure can be
recovered when the biorthogonal sets are no longer Riesz bases. In
particular, we will see under what conditions some relevant
operators we are going to introduce are, in fact, densely defined.

\section{Some examples when orthonormality is lost}\label{sect3}

In this section we discuss how, and to which extent, loosing orthonormality  can give rise to certain mathematical and physical consequences which do not appear whenever one uses o.n. bases.
Of course, we will also give up the assumption that the set of vectors we consider is a Riesz basis, since this case is completely under control.

For making these differences evident, we discuss in some details three examples, the last one being directly physically motivated, whereas the first two are interesting mainly from a mathematical point of view, but not only.

\subsection{First example}\label{sectFE}

Let $\F_e:=\{e_n,\,n\geq1\}$ be an o.n. basis of a Hilbert space $\Hil$, and let us introduce the set $\F_x=\{x_n=\sum_{k=1}^n\frac{1}{k}\,e_k,\,n\geq1\}$, see \cite{heil}. Of course we can write $x_n$ as follows: $x_1=e_1$, and $x_n=x_{n-1}+\frac{1}{n}\,e_n$, $n\geq2$.
It is clear that $x_n\in\Hil$ for all $n\geq1$. Also, $\F_x$ is complete in $\Hil$. In fact, $\left<f,x_n\right>=0$ for all $n\geq1$
implies that $\left<f,e_n\right>=0$ for all $n\geq1$ as well, so that $f=0$, necessarily. Let us now introduce $\G$ as the linear span of the $e_n$'s. This set is dense in $\Hil$. It is possible to see that $\F_x$ is a basis for $\G$, \cite{bagbook}, but not for $\Hil$. In particular, if on one hand it is easy to check that each vector $f=\sum_{k=1}^N c_ke_k$, $N<\infty$, can be written as a linear combination of the $x_n$'s, on the other hand it is also possible to check that $h:=\sum_{k=1}^\infty\frac{1}{k}\,e_k$,
which is a non zero vector in $\Hil$, cannot be written as $\sum_{k=1}^\infty\alpha_k\,x_k$, for any choice of the complex numbers $\alpha_k$.
In fact, assume that this is possible. Then,  we should have
$$\left\{
\begin{array}{ll}
\left<h,e_1\right>=1, \mbox{ and } \left<h,e_1\right>=\sum_{k=1}^\infty\alpha_k \,\,\qquad \Rightarrow\,\, \qquad \sum_{k=1}^\infty\alpha_k=1\\
\left<h,e_2\right>=\frac{1}{2}, \mbox{ and } \left<h,e_2\right>=\frac{1}{2} \sum_{k=2}^\infty\alpha_k \qquad \Rightarrow \qquad \sum_{k=2}^\infty\alpha_k=1\\
\left<h,e_3\right>=\frac{1}{3}, \mbox{ and } \left<h,e_3\right>=\frac{1}{3} \sum_{k=3}^\infty\alpha_k \qquad \Rightarrow \qquad \sum_{k=3}^\infty\alpha_k=1,\\
\end{array}%
\right.
$$
and so on. Hence, we should have $\alpha_1=\alpha_2=\alpha_3=\ldots=0$, which implies that $h=0$, which is absurd.

The set which is biorthogonal to $\F_x$ is
 $\F_y:=\{y_n=ne_n-(n+1)e_{n+1}, \,n\geq1\}$. Indeed we can prove that $\left<x_k,y_l\right>=\delta_{k,l}$, $\forall\,k,l\in\Bbb N$.
 $\F_y$ is  complete in $\G$, but not in $\Hil$. Indeed, let $f\in\G$ be orthogonal to all the $y_n$'s. The vector $f$ can be written as $f=\sum_{k=1}^Nc_k e_k$, for some finite $N$. Of course, $c_k=\left<e_k,f\right>$.
 Now, condition $\left<f,y_1\right>=0$ implies that $\left<f,e_1\right>=2\left<f,e_2\right>$. Also, from $\left<f,y_2\right>=0$,  it follows that $2\left<f,e_2\right>=3\left<f,e_3\right>$, and so on. However, since $\left<f,y_N\right>=0$, we deduce that  $N\left<f,e_N\right>=(N+1)\left<f,e_{N+1}\right>=0$. Then, $\left<f,e_k\right>=0$ for all $k=1,2,\ldots,N$, so that $f=0$. Therefore, as stated, $\F_y$ is complete in $\G$. To prove that $\F_y$ is not complete in $\Hil$, it is sufficient to observe that the  vector $h$, already introduced,  is orthogonal to all the $y_n$'s, but it is not zero.

Contrarily to $\F_x$, it can be shown that the set $\F_y$ is not a basis for $\G$.  Therefore, a fortiori, $\F_y$ is not a basis for $\Hil$. We can further prove that, even though $\F_x$ and $\F_y$ are not quasi-bases,  they are still $\G$-quasi bases. To prove these claims we first observe that, taking again $h$ as above, on one hand we have $\|h\|^2=\left<h,h\right>=\sum_{k=1}^\infty\frac{
 1}{k^2}=\frac{\pi^2}{6}$, whereas on the other hand we have
 $\sum_{n=1}^\infty \left<h,x_n\right>\left<y_n,h\right>=0$, since $\left<y_n,h\right>=0$ for all $n$. Hence, at least for this $h$, $\left<h,h\right>\neq \sum_{n=1}^\infty \left<h,x_n\right>\left<y_n,h\right>$, and our first assertion is proved. Of course, this is in agreement with the fact that $\F_y$ is not complete in $\Hil$, as it should be, if they were quasi-bases. However, they are $\G$-quasi bases because, taking $f$ and $g$ in $\G$,
 it is just a  straightforward computation to check that
$$
\sum_{n=1}^\infty \left<f,x_n\right>\left<y_n,g\right>=\sum_{n=1}^\infty \left<f,y_n\right>\left<x_n,g\right>=\left<f,g\right>,
$$
which is what we had to prove.

\subsection{Second example}\label{sectSE}
Let, as before, $\F_e:=\{e_n,\,n\geq1\}$ be an o.n. basis of a Hilbert space $\Hil$, $\G$  its linear span, and consider the sets

$$
\F_x=\left\{x_n=\sum_{k=1}^n(-1)^{n+k}e_k,\,n\geq1\right\} \quad\mbox{ and } \quad \F_y=\{y_n=e_n+e_{n+1},\,n\geq1\}.
$$
Then $\left<y_n,x_k\right>=\delta_{n,k}$ for all $k,n\geq1$, \cite{chri}.
$\F_x$ is complete in $\Hil$ and, interestingly enough, it is a basis for $\G$.  Indeed, let $f\in\G$, then $f$ can be written as a finite linear combination of the $e_n$'s. Let us assume, to begin with, that $f=\sum_{k=1}^{2N}c_ke_k$.
 Now we will show that there exist $\alpha_j\in\Bbb C$, $j=1,2,\ldots,2N$ such that $f=\sum_{k=1}^{2N}\alpha_kx_k$. In fact, equating the two expansions, we deduce that
$$
\underline c_{ 2N}=T_{2N} \underline\alpha_{ 2N},
$$
where
$$
\underline c_{ 2N}=\left(
                     \begin{array}{c}
                       c_1 \\
                       c_2 \\
                       c_3 \\
                       . \\
                       . \\
                       c_{2N-1} \\
                       c_{2N} \\
                     \end{array}
                   \right),\quad
                   \underline \alpha_{ 2N}=\left(
                     \begin{array}{c}
                       \alpha_1 \\
                       \alpha_2 \\
                       \alpha_3 \\
                       . \\
                       . \\
                       \alpha_{2N-1} \\
                       \alpha_{2N} \\
                     \end{array}
                   \right),\quad    T_{2N}=\left(
                                             \begin{array}{ccccccc}
                                               1 & -1 & 1 & . & . & 1 & -1 \\
                                               0 & 1 & -1 & . & . & -1 & 1 \\
                                               0 & 0 & 1 & . & . & 1 & -1 \\
                                               . & . & . & . & . & . & . \\
                                               . & . & . & . & . & . & . \\
                                               0 & 0 & 0 & . & . & 1 & -1 \\
                                               0 & 0 & 0 & 0 & 0 & 0 & 1 \\
                                             \end{array}
                                           \right).
$$

\vspace{2mm}

Since $\det(T_{2N})=1$ for all $N$, $T_{2N}^{-1}$ surely exists, and therefore $\underline\alpha_{ 2N}=T^{-1}_{2N}\underline c_{ 2N}$.
Then we recover the coefficients $\alpha_n$'s of the expansion of $f$ in terms of $x_n$'s. Exactly the same conclusion we find if $f=\sum_{k=1}^{2N+1}c_ke_k$.
Hence, $\F_x$ is a basis for $\G$, as stated. What is more, one also can  prove that $\F_x$ is, in fact, a basis also for $\Hil$.
This is not particularly surprising, since the determinant of $T_{2N}$ (and of $T_{2N+1}$), i.e. the possibility of inverting those matrices, is independent of $N$.
{A consequence is the completeness of $\F_x$ in $\Hil$.}
As for
  $\F_y$, this set is complete in $\Hil$.
 Indeed if  $f\in\Hil$ is orthogonal to every $y_n$, then we easily deduce that $\left|\left<f,e_j\right>\right|$ is independent of $j$ so that
 $\|f\|^2=\sum_{j=1}^\infty\left|\left<f,e_j\right>\right|^2$ can be finite { if} only if $\left<f,e_j\right>=0$ for all $j$.
  Hence $f=0$. However, it turns out that $\F_y$ is not a basis for $\Hil$, \cite{chri}.
  For instance, one can notice that $e_1$ cannot be written in terms of $y_n$'s.
  This incidentally also implies  that, as in the previous example, $\F_y$ cannot even  be a basis for $\G$, since $e_1\in\G$.

It is now interesting to see that, even though $\F_y$ is not a basis for $\G$, $\F_x$ and $\F_y$ are  $\G$-quasi bases. This can be proved with a direct computation:
$$
\left<f,g\right>=\sum_{n=1}^\infty \left<f,y_n\right>\left<x_n,g\right>=\sum_{n=1}^\infty \left<f,x_n\right>\left<y_n,g\right>=\sum_{n=1}^{Min(N,M)}\overline{f_n}\,g_n,
$$
for all $f,g\in\G$ such that $f=\sum_{n=1}^{N}f_ne_n$, $g=\sum_{n=1}^{M}g_ne_n$ with $f_i,g_j\in\mathbb{C}$, $i=1,\cdots,N$ and $j=1,\cdots, M$. It is evident, therefore, that as in Section \ref{sectFE}, also here it is possible to recover a (weak) resolution of the identity, even though we are working with biorthogonal sets which are not bases.

\subsection{Third example, with Hamiltonians}\label{sectTE}

This example  is, in a certain sense, more  physically-motivated, since it is directly linked to a quantum harmonic oscillator. Moreover, exactly for this reason, it is also relevant because it will suggest how to enrich our previous examples by adding some physical insight to their original mathematical aspects. We start defining the following functions of $\Sc(\Bbb R)$, the set of $C^\infty$, fast decreasing, functions:
$$
x_n(x)=\frac{1}{\sqrt{2^n\,n!\sqrt{\pi}}}\,H_n(x)\,e^{-\frac{x^2}{4}},\qquad
y_n(x)=\frac{1}{\sqrt{2^n\,n!\sqrt{\pi}}}\,H_n(x)\,e^{-\frac{3x^2}{4}},
$$
for $n=0,1,2,3,\ldots$. Here $H_n(x)$ is the $n$-th Hermite
polynomial. It is easy to check that
$$\left<x_m,y_n\right>=\left<e_m,e_n\right>=\delta_{m,n},$$ where
$e_n(x)=\frac{1}{\sqrt{2^n\,n!\sqrt{\pi}}}\,H_n(x)\,e^{-\frac{x^2}{2}}$
is the $n$-th function of the set $\F_e=\{e_n(x)\}$, which is the
well known o.n. basis of $\Lc^2(\Bbb R)$ consisting of eigenvectors
of the self-adjoint Hamiltonian of the
quantum harmonic oscillator, \cite{messiah}.  We continue to call $\G$ the linear span of the $e_n$'s.
We have
$$
h\,e_n=\left(n+\frac{1}{2}\right)e_n,
$$
$n=0,1,2,3,\ldots$. It is clear that, for all such $n$'s,
$y_n=Te_n$, $x_n=T^{-1}e_n$, where $T$ is the multiplication
operator defined as $(Tf)(x)=e^{-\frac{x^2}{4}}f(x)$, for all
$f(x)\in\Lc^2(\Bbb R)$. Of course, $T$ is bounded and self-adjoint.
It is also invertible, but its inverse is unbounded. However,
$T^{-1}$ is densely defined since its domain, $D(T^{-1})$, contains e.g. the set $D(\Bbb R)$ of the compactly supported
$C^\infty$-functions, which is dense in $\Lc^2(\Bbb R)$. Of course, this is a proper inclusion since each $e_n(x)$  belongs to
$D(T^{-1})$, but does not belong to $D(\Bbb R)$, for any $n$.

A standard argument, see \cite{kol}, shows that $\F_y=\{y_n(x), \,n\geq0\}$
and $\F_x=\{x_n(x), \,n\geq0\}$ are both complete in $\Lc^2(\Bbb R)$. They are also
$D(\Bbb R)$-quasi bases: in fact, let $f,g\in D(\Bbb R)$, then
$$
\left<f,g\right>=\left<T^{-1}Tf,g\right>=\left<Tf,T^{-1}g\right>=\sum_{n=0}^\infty\left<Tf,e_n\right>\left<e_n,T^{-1}g\right>=
$$
$$
=\sum_{n=0}^\infty\left<f,Te_n\right>\left<T^{-1}e_n,g\right>=\sum_{n=0}^\infty\left<f,y_n\right>\left<x_n,g\right>.
$$
Analogously we can check that
$\left<f,g\right>=\sum_{n=0}^\infty\left<f,x_n\right>\left<y_n,g\right>$.
On the other hand, $\F_y$ and $\F_x$ are not bases for $\Lc^2(\Bbb
R)$ because,  {as we have already recalled in Remark \ref{rem proj}, a necessary condition for $\F_y$ and $\F_x$ to be bases is that}  $\sup_n\|y_n\|\|x_n\|<\infty$,  but this is
not the case. In fact, (see the integral 2.20.16. nr. 2 in \cite{prud}), we get
$$
\|y_n\|^2=\sqrt{\frac{2}{3}}\,\frac{2}{3^{n/2}}\,P_n\left(\frac{2}{\sqrt{3}}\right),
\qquad
\|x_n\|^2=\sqrt{2}\,3^{n/2}\,P_n\left(\frac{2}{\sqrt{3}}\right),
$$
where $P_n(x)$ is the $n$-th Legendre polynomial. Now, using the asymptotic behavior in $n$ of these polynomials, for $x>1$, see \cite{szego},  we see that, for large $n$,
$$
\|y_n\|^2\|x_n\|^2\simeq \frac{2}{\sqrt{3}\,\pi}\,\frac{3^n}{n},
$$
which diverges with $n$. Hence $\sup_n\|y_n\|\|x_n\|=\infty$. Therefore, neither $\F_y$ nor $\F_x$ can be bases for $\Lc^2(\Bbb R)$. Nevertheless, they have interesting physical properties, since they are eigenstates of the following manifestly non self-adjoint operators $$H_1=\frac{1}{2}\left[-\frac{d^2}{dx^2}-x\frac{d}{dx}+\frac{1}{2}\left(\frac{3x^2}{2}-1\right)\right]$$ and
$$H_2=\frac{1}{2}\left[-\frac{d^2}{dx^2}+x\frac{d}{dx}+\frac{1}{2}\left(\frac{3x^2}{2}+1\right)\right].$$
In fact, since the Hermite polynomial $H_n(x)$ satisfies the differential equation $H_n''(x)-2xH_n'(x)+2nH_n(x)=0$, for all $n\in \mathbb{N}$, a direct computation shows that
$$
H_1\,y_n=E_ny_n,\qquad H_2\,x_n=E_nx_n,
$$
where $E_n=n+\frac{1}{2}$ and $n=0,1,2,3,\ldots$.

Then the three Hamiltonians $H_1$, $H_2$ and $h$ are all isospectrals. This is in agreement with the fact that they are all ({\em strongly}) similar, i.e. that, for every $f\in D(\Bbb R)$, we have $ThT^{-1}f=H_1f$ and $T^{-1}hTf=H_2f$.
This also suggests that, for all those functions, $H_2f=H_1^\dagger f$, as one also can explicitly check.

Quite often, when isospectral Hamiltonians appear related by some
(possibly extended, as in this case) similarity operator, it is a
standard procedure to introduce the following operators:
$$
D(S_x)=\left\{g\in\Hil:\,\sum_{n=0}^\infty\left<x_n,g\right>x_n
\,\mbox{ exists in }\Hil\right\}, \, D(S_y)=\left\{f\in\Hil:\,\sum_{n=0}^\infty\left<y_n,f\right>y_n \,\mbox{ exists in
}\Hil\right\},
$$
and $S_x
g=\sum_{n=0}^\infty\left<x_n,g\right>x_n$, $S_y f=\sum_{n=0}^\infty\left<y_n,f\right>y_n$, for $g\in
D(S_x)$ and $f\in D(S_y)$. Using the continuity of $T$, it is clear that $S_y$ is
everywhere defined and that $S_y=T^2$. In fact, taken $f\in D(S_y)$,
we have
$$
S_y
f=\sum_{n=0}^\infty\left<y_n,f\right>y_n=\sum_{n=0}^\infty\left<Te_n,f\right>Te_n=T\left(\sum_{n=0}^\infty\left<e_n,T\,f\right>e_n\right)=T(Tf).
$$

Of course, since $T^2$ is bounded, this equality can be extended to the whole $\Hil$. Notice that, in particular, $S_yx_n=y_n$ for all $n$.

Of course, such a simple argument does not hold for $S_x$. This is
because $T^{-1}$ is unbounded and, therefore, not continuous on $\Hil$.
However, we can still prove that $S_x=T^{-2}$, but in a {\em weak
form}, i.e. we can prove that
$$
\left<\xi,\left(T^{-2}-S_x\right)g\right>=0,
$$
for all $\xi\in D(\Bbb R)$, and for all $g\in D(\Bbb R)\cap D(S_x)$,
which we assume here to be a sufficiently rich set. This is a
reasonable assumption since, recalling that $S_xy_n=x_n$ for all
$n$, we see that the linear span of the $y_n$'s, $\D_y$, is a subset
of $D(S_x)$. Moreover, we observe that $\D_y$ is the image of the
dense set $\G$, via the bounded operator $T$. Since $\F_y$ is
complete, its linear span $\D_y$ is dense in $\Hil$.  So, $D(S_x)$
is surely a rather rich set. Of course, what is not evident is that
$D(\Bbb R)\cap D(S_x)$ is also rich, and will only be assumed here.

\section{Hamiltonians defined by $\G$-quasi bases: theory and examples}

The examples  in literature, together with those introduced in
Section \ref{sect3}, show that there exist several biorthogonal sets
of vectors with different characteristics, most of which are harder
to deal with than o.n bases, but which are still interesting (both
in mathematics and in physics) and sufficiently well behaved. In
what follows, somehow inspired by what we have done in Section
\ref{sectTE}, we discuss more mathematical and physical facts
related to the biorthogonal sets considered in Sections \ref{sectFE}
and \ref{sectSE}.

In particular, we use  the general results on $\F_x$ and $\F_y$ deduced in those sections to define some  manifestly non self-adjoint operators, which we still call Hamiltonians,  having $x_n$ and $y_n$ as eigenstates. In this way we will significantly extend what was first proposed in \cite{bit}, starting from biorthogonal Riesz bases and then in \cite{fg,bellom} in the more general settings of rigged Hilbert spaces. Interestingly enough, we will see that  many of the results deduced in \cite{bit} also can  be recovered here, in a situation in which $\F_x$ and $\F_y$ are not even bases, but just $\G$-quasi bases, with $\G$ some dense subset of $\Hil$.

To make our results model-independent, we devote the first part of this section to discuss some general results which extend those in \cite{bit} (see also Section \ref{sectRB}) to the present settings. In the remaining part of the section, we will go back to the particular choices considered in Sections \ref{sectFE} and \ref{sectSE}, and we will see what can be said in those cases.

\subsection{Some general results}\label{sectSGR}

Let ${\balpha}:=\{\alpha_n, \,n\in \Bbb N\}$ be a sequence of
complex numbers, $\F_x$ and $\F_y$ biorthogonal $\G$-quasi bases for some dense subset $\G\subset\Hil$, and $H_{x, y}^\balpha  $ and $H_{y,x}^\balpha  $ two operators defined
as follows:
$$D(H_{x, y}^\balpha  )=\left\{f\in\Hil:\,  \sum_{n=1}^\infty \alpha_n\left<y_n,f\right>\,x_n \mbox{ exists in }\Hil\right\},$$
 $$D(H_{y,x}^\balpha  )=\left\{g\in\Hil:\, \sum_{n=1}^\infty \alpha_n\left<x_n,g\right>\,y_n\mbox{ exists in }\Hil\right\},$$
and
\begin{equation}\label{23}
H_{x, y}^\balpha   f:=\sum_{n=1}^\infty \alpha_n\left<y_n,f\right>\,x_n, \qquad
H_{y,x}^\balpha   g:=\sum_{n=1}^\infty \alpha_n\left<x_n,g\right>\,y_n,
\end{equation}
for all $f\in D(H_{x, y}^\balpha  )$ and $g\in D(H_{y,x}^\balpha)$. In analogy with what has been discussed in Section \ref{sectRB}, we easily see that

\begin{align}\label{2.1} &\D_y:= \mbox{span}\{y_n\} \subseteq D(H_{y,x}^\balpha  )\quad
 \D_x:= \mbox{span}\{x_n\} \subseteq D(H_{x, y}^\balpha  );\\
& H_{y,x}^\balpha   y_k =\alpha_k y_k, \; \hspace{2.5cm}
 H_{x, y}^\balpha   x_k =\alpha_k x_k, \quad k\geq1. \label{2.2}
\end{align}
 Therefore, the $x_n$'s
and the $y_n$'s are eigenstates respectively of $ H_{x, y}^\balpha$
and $ H_{y, x}^\balpha$, and the complex numbers $\alpha_n$'s are
their (common) eigenvalues. So, from this point of view, not much
has changed with respect to what we have summarized in Section
\ref{sectRB}. What is really different here is that, since neither
$\F_x$ nor $\F_y$ are (Riesz) bases, neither $H_{y, x}^\balpha $ nor
$H_{x, y}^\balpha  $ need to be densely defined, in general.
However, when this is true, more can be deduced\footnote{It is clear that having densely defined Hamiltonians does not necessarily imply that their eigenvectors do form  Riesz bases. In fact, in Section \ref{sectTE} we have seen an explicit example in which $H_1$ and $H_2$ are densely defined even though their eigenstates are not Riesz bases.}.  We will assume, in the
remaining part of this section, that the operators  $H_{y,
x}^\balpha$ and $H_{x, y}^\balpha  $ are densely defined. For this reason we will
see that suitable conditions exist which make these assumptions
verified in concrete situations, as the examples in Sections
\ref{Sect: BFE} and \ref{bts} show.

\begin{prop}
Let $\{\alpha_n\}\subset \Bbb R$, then $(H_{y,x}^\balpha)^\dagger\supseteq
H_{x, y}^\balpha  $.\end{prop}\begin{proof} We have to check that each $h\in
D(H_{x, y}^\balpha )$ also belongs to $D((H_{y,x}^\balpha )^\dagger)$, and that for such
$h$'s, $(H_{y,x}^\balpha)^\dagger h=H_{x, y}^\balpha h$. If $h\in
D(H_{x, y}^\balpha )$, the series $\sum_{n=1}^\infty
\alpha_n\left<y_n,h\right>\,x_n$ is norm convergent to $H_{x, y}^\balpha h$.
Hence, using the continuity of the scalar product, we have
$$
\left<f,H_{x, y}^\balpha h\right>=\left<f,\sum_{n=1}^\infty
\alpha_n\left<y_n,h\right>\,x_n\right>=\sum_{n=1}^\infty
\alpha_n\left<y_n,h\right>\left<f,x_n\right>=
$$
$$
=\left<\sum_{n=1}^\infty
\alpha_n\left<x_n,f\right>\,y_n,h\right>=\left<H_{y,x}^\balpha f,h\right>,
$$
for all $f\in D(H_{y,x}^\balpha )$. This means that $h\in D((H_{y,x}^\balpha)^\dagger)$ and that $(H_{y,x}^\balpha)^\dagger h=H_{x, y}^\balpha  h$, which is what we had to prove.
\end{proof}

\begin{rem} If $\{\alpha_n\}\subset \Bbb R$, in some cases the two operators $(H_{y,x}^\balpha)^\dagger$ and $H_{x, y}^\balpha $ do coincide. One of these cases is when the sets $\F_x$ and $\F_y$ are Riesz bases, see e.g. \cite{bit}. Another, even simpler, situation is when $H_{x, y}^\balpha $ and $H_{y,x}^\balpha $ are bounded operators, which is the case, for instance, if the sequence $\{|\alpha_n|\|x_n\|\|y_n\|\}$ belongs to $l^1(\Bbb R)$.
\end{rem}
%%%%%%%%%%%%%%%%%%%%%%%%%%%%%%%%%%%%%

Let us define now, in analogy e.g. with \cite{bit}, the following lowering and raising operators which can be used
to factorize the  Hamiltonians introduced before. To this aim, we assume that the
sequence $\balpha$ satisfies the following condition: $0=\alpha_1\leq\alpha_2\leq\ldots$. Then we introduce the following operators:
$$ \left\{\begin{array}{l} D(A_{x,y})=\left\{f\in \Hil; \sum_{n=2}^\infty
\sqrt{\alpha_n}\ip{y_n}{f}x_{n-1} \mbox{ exists in }\Hil \right\}\\
{ } \\
A_{x,y} f= \sum_{n=2}^\infty \sqrt{\alpha_n}\ip{y_n}{f}x_{n-1},
\qquad f \in D(A_{x,y})
\end{array}\right.
$$

$$ \left\{\begin{array}{l} D(A_{y,x})=\left\{f  \in \Hil; \sum_{n=2}^\infty \sqrt{\alpha_n}{\ip{x_n}{f }}y_{n-1} \mbox{ exists in }\Hil \right\}\\
{ } \\
A_{y,x}f = \sum_{n=2}^\infty \sqrt{\alpha_n}{\ip{x_n}{f }}y_{n-1},
\qquad f  \in D(A_{y,x})
\end{array}\right.
$$

$$ \left\{\begin{array}{l} D(B_{x,y} )=\left\{f\in \Hil; \sum_{n=1}^\infty \sqrt{\alpha_{n+1}}\ip{y_n}{f}x_{n+1} \mbox{ exists in }\Hil \right\}\\
{ } \\
B_{x,y} f= \sum_{n=1}^\infty \sqrt{\alpha_{n+1}}\ip{y_n}{f}x_{n+1},
\qquad f \in D(B_{x,y})
\end{array}\right.
$$

$$ \left\{\begin{array}{l} D(B_{y,x})=\left\{f  \in \Hil; \sum_{n=1}^\infty
\sqrt{\alpha_{n+1}}\ip{x_n}{f }y_{n+1}\mbox{ exists in }\Hil \right\}\\
{ } \\
B_{y,x}f = \sum_{n=1}^\infty \sqrt{\alpha_{n+1}}\ip{x_n}{f }y_{n+1},
\qquad f  \in D(B_{y,x}).
\end{array}\right.
$$

These operators behave as some sort of ladder operators. In fact we have, for instance: $x_n\in D(A_{x,y})$ and $A_{x,y}x_n=\sqrt{\alpha_n}x_{n-1}$, for all $n\geq2$. Also, $x_n\in D(B_{x,y})$ and $B_{x,y}x_n=\sqrt{\alpha_{n+1}}x_{n+1}$, for all $n\geq1$, and so on.

From now on we will assume that (at least) $\F_x$ is a basis for
some dense subset $\G\subseteq\Hil$. This hypothesis holds true, for
instance, in Sections \ref{sectFE} and \ref{sectSE}, where, we
recall, $\G$ is the linear span of the vectors $e_n$ of a given o.n.
basis. Hence  $A_{x,y}$, $B_{x,y}$ and $H_{x, y}^\balpha $ are all densely defined.  Of course, similar results also can be  deduced  for
$A_{y,x}$, $B_{y,x}$ and for $H_{y,x}^\balpha $, at least if $\F_y$
is a basis for $\G$. However, this is not what happens in our
examples, unless we impose some extra conditions to $\balpha$, as we
will see later. When these extra conditions are satisfied, all our
operators turn out to be densely defined.

 Now we can prove the following Proposition, which allows us to factorize both $H_{x, y}^\balpha $ and $H_{y,x}^\balpha $ respectively on $\D_x$ and $\D_y$.

\begin{prop}\label{propfact} Let $\F_x$ and $\F_y$ be biorthogonal sets. Assume that the sequence
$\balpha$ satisfies the following condition:
$0=\alpha_1\leq\alpha_2\leq\ldots$, then the following statements
hold.
\begin{itemize}
  \item[$i)$] The operators $H_{x, y}^\balpha $ and $B_{x,y}A_{x,y}$
coincide on $\D_x$;
  \item[$ii)$] the operators $H_{y,x}^\balpha $ and $B_{y,x}A_{y,x}$ coincide on $\D_y$.
\end{itemize}
\end{prop}
\begin{proof}
$i)$ Let $f\in\D_x$. For every $ m\in\mathbb{N}$  we have
$$\ip{y_m}{A_{x,y}f}=\ip{y_m}{\sum_{n=2}^\infty \sqrt{\alpha_n}\ip{y_n}{f}x_{n-1}}=\sum_{n=2}^\infty \sqrt{\alpha_n}\ip{y_n}{f}\ip{y_m}{x_{n-1}}=\sqrt{\alpha_{m+1}}\ip{y_{m+1}}{f}$$ because of the continuity of the inner product. Then, recalling that $\alpha_1=0$, we have
\begin{eqnarray*}
  B_{x,y}(A_{x,y}f) &=& \sum_{m=1}^\infty \sqrt{\alpha_{m+1}}\ip{y_m}{A_{x,y}f}x_{m+1}= \\
   &=& \sum_{m=1}^\infty \sqrt{\alpha_{m+1}}\sqrt{\alpha_{m+1}}\ip{y_{m+1}}{f}x_{m+1}= \\
   &=&\sum_{m=1}^\infty\alpha_{m+1}\ip{y_{m+1}}{f}x_{m+1}= \sum_{n=2}^\infty\alpha_{n}\ip{y_{n}}{f}x_{n}\\
   &=&\sum_{n=1}^\infty\alpha_{n}\ip{y_{n}}{f}x_{n}=H_{x, y}^\balpha f,
\end{eqnarray*}
which is what we had to prove. Of course, our assertion  $ii)$ can be proved in the same way.

\end{proof}

Of course, factorizability of the Hamiltonians could be true not
just on  $\D_x$ and $\D_y$, but also on larger sets. In other words,
Proposition \ref{propfact} does not exclude that, for instance,
$H_{x, y}^\balpha \hat f=B_{x,y}A_{x,y}\hat f$ for some $\hat f$ belonging to
$\Hil\setminus\D_x$. This is the case when $D(H_{x, y}^\balpha )\supset \D_x$,
and when the set of vectors $\hat f$ such that $A_{x,y}\hat f\in
D(B_{x,y})$ is  larger than $\D_x$ too.

\begin{rem}
The possibility of factorizing the Hamiltonian of a physical system is quite useful in concrete applications, both for general reasons, and in connection with pseudo-hermitian and with supersymmetric quantum mechanics (SUSY-QM). This, in fact, simplifies the computation of the eigenstates and also  can produce more exactly solvable models, i.e. quantum mechanical Hamiltonians with known eigenvalues and eigenvectors. In fact, if $H$ can be written (at least formally) in a factorized form $H=BA$, using SUSY-QM we can deduce how and when the eigenvectors and the eigenvalues of the new Hamiltonian $H_1:=AB$ can be found, \cite{susy1,susy2}. Sometimes it happens that the procedure can be iterated, and in this case one is able to deduce a full family of solvable models.  Moreover, the same operators used to factorize the Hamiltonians are often used in connection with bi-coherent states, \cite{bag0}.
\end{rem}

\medskip

As in Section \ref{sectRB}, we can further introduce two sequences of
strictly positive real numbers ${\bbeta}:=\{\beta_n>0, \,n\in
\Bbb N\}$ and ${ \bgamma}:=\{\gamma_n>0, \,n\in \Bbb N\}$, and two related
operators $S_x^{\bbeta}$ and $S_y^{\bgamma}$ as follows:

$$D(S_x^{\bbeta})=\left\{f\in\Hil:\, \sum_{n=1}^\infty \beta_n\left<x_n,f\right>\,x_n\in\Hil\right\},\quad D(S_y^{\bgamma})=\left\{g\in\Hil:\, \sum_{n=1}^\infty\gamma_n \left<y_n,g\right>\,y_n\in\Hil\right\},$$
and \be S_x^{\bbeta}f=\sum_{n=1}^\infty
\beta_n\left<x_n,f\right>\,x_n, \qquad S_y^{\bgamma}g=\sum_{n=1}^\infty \gamma_n\left<y_n,g\right>\,y_n,
\label{24}\en for all $f\in D(S_x^{\bbeta})$ and $g\in D(S_y^{\bgamma})$. These operators are positive and, if densely defined, are
symmetric too. For instance, this is true if $\F_x$ and $\F_y$ are bases for $\G$, or when some suitable conditions on $\bbeta$ or $\bgamma$ are satisfied {as we will see in the following}. Whenever $D(S_x^{\bbeta})$ and $D(S_y^{\bgamma})$ are dense, both $S_x^{\bbeta}$ and $S_y^{\bgamma}$ admit a
self-adjoint (Friedrichs) extension, see \cite{peder}. Definition (\ref{24}) and the biorthogonality of the families $\F_x$ and $\F_y$ imply that, for all $n$, $y_n\in D(S_x^{\bbeta})$, $x_n\in D(S_y^{\bgamma})$, $S_x^{\bbeta}y_n=\beta_nx_n$ and  $S_y^{\bgamma}x_n=\gamma_ny_n$. Hence these operators map $\F_x$ into $\F_y$ and viceversa, with some extra normalization factor which we cannot get rid of.

Let us now see in details what happens when we consider the sets of vectors introduced in Sections \ref{sectFE} and \ref{sectSE}.

\subsection{Back to the first example}\label{Sect: BFE}
As we have shown in Section \ref{sectFE}, each vector of $\G$ can be written as a finite linear combination of the $x_n$'s. Then, being $\G$ dense in $\Hil$, $H_{x, y}^\balpha $ is densely defined. In fact, from (\ref{23}), we see that $\G\subseteq D(H_{x, y}^\balpha )$. On the other hand, in general we cannot say, using the same argument, that $\G$ is also contained in $D(H_{y,x}^\balpha )$, since $\F_y$ is not a basis for $\G$. Therefore, we cannot conclude that  $H_{y,x}^\balpha $ is densely defined, in general. However, it is possible to prove that, if ${\balpha}$ is such that $\{n\alpha_n\}$ belongs to the Hilbert space $l^2(\Bbb N)$, i.e. if $\sum_{n=1}^\infty n^2|\alpha_n|^2<\infty$, then $\G\subseteq D(H_{y,x}^\balpha )$, so that $H_{y,x}^\balpha $ is densely defined too. In fact, let $f\in\G$, then $f$ can be written as a finite linear combination $f=\sum_{l=1}^Mc_le_l$, for some $M$ with $c_l=\left<e_l,f\right>$. Then, after few computations,
$$
H_{y,x}^\balpha f=\sum_{l=1}^M\frac{c_l}{l}\sum_{n=l}^\infty \alpha_n y_n.
$$
It is clear that the series on the right-hand side converges if and only if $\sum_{l=1}^M\frac{c_l}{l}\sum_{n=1}^\infty \alpha_n y_n=:\tilde c_f \sum_{n=1}^\infty \alpha_n y_n$ converges, since the two series differ for a finite number of terms. Here we have introduced $\tilde c_f=\sum_{l=1}^M\frac{c_l}{l}$, which is clearly well defined. Then $f\in D(H_{y,x}^\balpha )$ if $\sum_{n=1}^\infty \alpha_n y_n$ converges in $\Hil$. Recalling now that for every $n\in\mathbb{N}$, $y_n=ne_n-(n+1)e_{n+1}$, one can check that
$$
\left\|\sum_{n=1}^\infty \alpha_n y_n\right\|^2=\sum_{n=1}^\infty |\alpha_n|^2\left(n^2+(n+1)^2\right)-\sum_{n=1}^\infty \left((n+1)^2 \overline{\alpha_n}\,\alpha_{n+1}+c.c.\right).
$$
Here $c.c.$ stands for complex conjugate. Using now our assumption on ${\balpha}$, and the Schwartz inequality, we conclude that $\|\sum_{n=1}^\infty \alpha_n y_n\|$ is finite, which is what we had to prove. From now on we will assume that $\balpha$ is such that $\{n\alpha_n\}\in l^2(\Bbb N)$.

Let us now see what can be said for the operators $S_y^{\bgamma}$ and $S_x^{\bbeta}$. Because of the properties of $\F_x$, $S_y^{\bgamma}$ is densely defined, symmetric and positive for any possible choice of $\bgamma=\{\gamma_n\}$, with $\gamma_n>0$. In contrast, the fact that each $y_n$ belongs to $D(S_x^{\bbeta})$ does not ensure us that $S_x^{\bbeta}$ is densely defined as well, in general. However, in analogy with what we have done for $H_{y,x}^\balpha $, we can see that, taken $f=\sum_{l=1}^Mc_le_l\in \G$, then
$ S_x^{\bbeta}f=\sum_{l=1}^M\frac{c_l}{l}\sum_{n=l}^\infty
\beta_n x_n. $ It is clear that the series on the right-hand side
converges if and only if $\sum_{n=1}^\infty \beta_n x_n$ converges.
Now, since for every $n\in \mathbb{N}$,
$\|x_n\|=\left(\sum_{k=1}^n\frac{1}{k^2}\right)^{1/2}<\left(\frac{\pi^2}{6}\right)^{1/2}$,
$$
\left\|\sum_{n=1}^\infty \beta_n x_n\right\|\leq \sum_{n=1}^\infty \beta_n\|x_n\|<\frac{\pi}{\sqrt{6}}\sum_{n=1}^\infty \beta_n,
$$ which is convergent if $\bbeta\in l^1(\Bbb N)$. Then, when this happens, $\G\subseteq D(S_x^{\bbeta})$, and therefore $S_x^{\bbeta}$ is also densely defined.
Of course, taking $\beta_n=1$  for every $n\in\mathbb{N}$ (as it is sometimes found in the literature, see \cite{bagbook,bit,bellom} for instance), does not appear to be a
good choice here, since in this case $D(S_x^\bbeta)$ is not a dense set, in principle. So we will not make this choice.
Now, if we  fix $\gamma_n=\frac{1}{\beta_n}$, we deduce that
$$
S_x^{\bbeta}S_y^{\bgamma}x_n=x_n,\quad\mbox{ and }\quad S_y^{\bgamma}S_x^{\bbeta}y_n=y_n,
$$
for all $n\in\Bbb N$. This, of course, does not imply that
$S_x^{\bbeta}$ is the inverse of $S_y^{\bgamma}$. In fact,
neither $\F_x$ nor $\F_y$ are bases for $\Hil$, so that the above
equalities cannot both be extended automatically to the whole Hilbert space. Nevertheless, they can be extended on some large sets, i.e. on $\D_x$ and on $\D_y$, which are in fact rather rich sets.

As in {\cite[Proposition 2.3]{bit}}, these operators
produce interesting intertwining relations:
\begin{equation}\label{two equalities}
\left(H_{x, y}^\balpha S_x^\bbeta-S_x^\bbeta H_{y,x}^\balpha \right)y_n=0,\qquad
\left(H_{y,x}^\balpha S_y^\bgamma-S_y^\bgamma H_{x, y}^\balpha \right)x_n=0,
\end{equation}
for all $n\in\Bbb N$, which are related to the fact that $H_{x, y}^\balpha $ and $H_{y,x}^\balpha $ share the same eigenvalues, and that $S_x^\bbeta$ and $S_y^\bgamma$ map $\F_x$ into (multiples of) $\F_y$ and viceversa.
As before, the second equality in \eqref{two equalities} can be extended to all vectors of $\G$, while the first one, with no further extra assumption,  cannot.
Finally, as stated at the end of Section \ref{sectSGR}, due to the fact that $S_x^{\bbeta}$ and $S_y^{\bgamma}$ are
positive and symmetric, they admit  self-adjoint extensions which we still indicate with the same symbols, and which are also positive. Hence, they admit positive square roots, which can be used to introduce,
at least formally,

$$\hat
e_n:=\frac{1}{\sqrt{\beta_n}}\left(S_x^{\bbeta}\right)^{1/2}y_n,
\qquad
h_{x,y}=\left(S_y^{\bgamma}\right)^{1/2}H_{x, y}^\balpha \left(S_x^{\bbeta}\right)^{1/2}.$$
Then we can, again formally, check that $\hat
e_n:=\frac{1}{\sqrt{\gamma_n}}\left(S_y^{\bgamma}\right)^{1/2}x_n$
and that $h_{x,y}\hat e_n=\alpha_n\hat e_n$. This shows that
$h_{x,y}$ has the same eigenvalues $\{\alpha_n\}$ as $H_{x, y}^\balpha $, for
instance. It worths to stress that these last  claims are based on
some subtle mathematical assumptions, which are not necessarily
satisfied, in general, if $\F_x$ and $\F_y$ are not Riesz or o.n.
bases. We refer to \cite{bit} for more details on this kind of
problems in the easiest situation, i.e. when $\F_x$ and $\F_y$ are
Riesz bases indeed. Here we just want to say that, while it is easy
to see that $\F_{\hat e}=\{\hat e_n\}$ is an  o.n. set, we can
conjecture that $\F_{\hat e}$ is not a basis for $\Hil$. In fact,
this set is the image of $\F_x$ (and $\F_y$), none of which is a
basis for $\Hil$.

\subsection{Back to the second example}\label{bts}

Similar considerations as those discussed in Section \ref{Sect: BFE}
can also  be worked out by considering the two sets introduced in
Section \ref{sectSE}, where we have shown, among other things, that
the $x_n$'s form a basis for $\G$ and for $\Hil$ too. Then, being
$\G$ dense in $\Hil$, $H_{x, y}^\balpha $ is densely defined. In fact, from
(\ref{23}), we see that $\G\subseteq D(H_{x, y}^\balpha )$. On the other hand,
in general we cannot say, using the same argument, that $\G$ is also
contained in $D(H_{y,x}^\balpha )$, since $\F_y$ is not a basis for $\G$.

%%%%
However, it is possible to prove that, if ${\balpha}$ is such that
$\{\alpha_n\}$ belongs to $l^2(\Bbb N)$, i.e. if
$\sum_{n=1}^\infty|\alpha_n|^2<\infty$, then $\G\subseteq
D(H_{y,x}^\balpha )$, so that  $H_{y,x}^\balpha $ is  densely defined too. In fact,
let $f\in\G$ and, as before, let $f=\sum_{l=1}^Mc_le_l$, for some
$M$, with $c_l=\left<e_l,f\right>\in\mathbb{C}$. Then

\begin{eqnarray*}
 H_{y,x}^\balpha f&=& H_{y,x}^\balpha \left(\sum_{k=1}^Mc_ke_k\right)=\sum_{k=1}^M c_k  H_{y,x}^\balpha e_k=\sum_{k=1}^M c_k\sum_{n=1}^\infty
\alpha_n
\ip{x_n}{e_k}y_n  \\
 &=&\sum_{k=1}^M c_k \sum_{n=k}^\infty (-1)^{n+k}\alpha_n y_n
\end{eqnarray*}

being
\begin{equation}\label{innpro}
\ip{x_n}{e_k}=(-1)^n\ip{-e_1+e_2+\cdots+(-1)^ne_n}{e_k}
=
\left\{
  \begin{array}{ll}
    (-1)^{n+k}, & \hbox{for $n\geq k$;} \\
    0, & \hbox{otherwise.}
  \end{array}
\right.
\end{equation}

Now,
\begin{eqnarray*}
   \|H_{y,x}^\balpha f\|&=& \left\|\sum_{k=1}^M c_k \sum_{n=k}^\infty (-1)^{n+k}\alpha_n
y_n\right\|=\left\|\sum_{k=1}^M (-1)^{k}c_k \sum_{n=k}^\infty
(-1)^{n}\alpha_ny_n\right\|\\& \leq& \sum_{k=1}^M |c_k|
\left\|\sum_{n=k}^\infty
(-1)^{n}\alpha_ny_n\right\|=L_f\left\|\sum_{n=k}^\infty
(-1)^{n}\alpha_n y_n\right\|\end{eqnarray*} with $L_f=\sum_{k=1}^M
|c_k|$. Since $\left\|\sum_{n=k}^\infty
(-1)^{n}\alpha_ny_n\right\|<\infty$ if and only if
$\left\|\sum_{n=1}^\infty (-1)^{n}\alpha_n y_n\right\|<\infty$, and
since $$\left\|\sum_{n=1}^\infty (-1)^{n}\alpha_n
y_n\right\|^2=2\left[\sum_{n=1}^\infty
|\alpha_n|^2-\Re\left(\sum_{n=1}^\infty
\alpha_n\overline{\alpha_{n+1}} \right)\right]$$ (where $\Re z$
stays for real part of $z$) after some calculations, we have:
$$\left\|\sum_{n=1}^\infty (-1)^{n}\alpha_n
y_n\right\|^2\leq3\sum_{n=1}^\infty |\alpha_n|^2$$ and the r.h.s.
converges since $\{\alpha_n\}\in l^2( \mathbb{N})$.

Consider now the operators defined as in \eqref{24}. Because of the
properties of $\F_x$, $S_y^\bgamma$ is densely defined, symmetric and
positive. On the other hand, the fact that each $y_n$ belongs to
$D(S_x^\bbeta)$ does not ensure us that $S_x^\bbeta$ is densely
defined as well, as we have already observed in the previous section. However, also in this case it is
possible to give some sufficient conditions in order  $S_x^\bbeta $ to be
densely defined. For example, it is enough to require that
   $\{\beta_n\sqrt{n}\}\in l^1(\mathbb{N})$. If this is the case, then, recalling  \eqref{innpro},

\begin{eqnarray*}
  \|S_x^\bbeta f\|  &=& \left\|S_x^\bbeta \left(\sum_{k=1}^Mc_k e_k\right)\right\| =
   \left\|\sum_{k=1}^Mc_k\sum_{n=1}^\infty\beta_n
   \left<x_n,e_k\right>\,x_n\right\| =\left\|\sum_{k=1}^Mc_k\sum_{n=k}^\infty\beta_n
   (-1)^{n+k}\,x_n\right\| \\
   &\leq& \sum_{k=1}^M|c_k|\left\|\sum_{n=k}^\infty\beta_n
   (-1)^n\,x_n\right\|,
\end{eqnarray*}
which converges if and only if $\sum_{k=1}^M|c_k|\left\|\sum_{n=1}^\infty\beta_n
   (-1)^n\,x_n\right\|= L_f\left\|\sum_{n=1}^\infty\beta_n
   (-1)^n\,x_n\right\|$ converges, which is trivially true if $\{\beta_n\sqrt{n}\}\in l^1(\mathbb{N})$, with  $L_f$ the same constant as before.

Here it is possible to repeat the same consideration as in Section \ref{Sect: BFE}:  the fact that $H_{y,x}^\balpha $ and $H_{x, y}^\balpha $ share the same eigenvalues,
and that $S_y^\bgamma$ and $S_x^\bbeta$ map $\F_y$ into (multiple of) $\F_x$ and viceversa is
reflected by the following weak form of the intertwining relations:
\begin{equation}\label{two more equalities}
\left(H_{y,x}^\balpha S_y^\bgamma-S_y^\bgamma H_{x, y}^\balpha \right)x_n=0,\qquad \left(H_{x, y}^\balpha S_x^\bbeta-S_x^\bbeta H_{y,x}^\balpha \right)y_n=0,
\end{equation}
for all $n\in\Bbb N$. Since $\F_x$ is a basis for $\G$ we can
conclude that the first equality in \eqref{two more equalities} holds on $\G$,  because every
vector $f\in \G$ can be written as a finite combination of the
$x_n$'s. In contrast, the second equality holds on the whole $\G$
only under additional assumptions.

\vspace{2mm}

\begin{rem}Going back to the example in Section \ref{sectTE},
it is clear that we can adopt the same representation for the
Hamiltonians $H_1$ and $H_2$ than in this section and write them as
in (\ref{23}). Quite often, see \cite{fg,bit,bellom}, one
adopts the more compact expressions:
$$
H_1=\sum_{n=0}^\infty E_n y_n\otimes \overline{x_n}, \qquad H_2=\sum_{n=0}^\infty E_n x_n\otimes \overline{y_n},
$$
where, for instance,  $\left(y_n\otimes
\overline{x_n}\right)(f)=\left<x_n,f\right>y_n$, for all $f\in\Hil$.
Of course, these equations must be completed with some information
on the domains of $H_1$ and $H_2$. For instance $D(H_1)=\{f\in\Hil:
H_1f \,\, \mbox{ exists  in } \Hil\}$. For what we have seen before,
this is surely dense in $\Lc^2(\Bbb R)$ since $D(\Bbb R)\subset
D(H_1)$. The same result can be shown for $H_2$.
\end{rem}

\section{Conclusions}
In this paper we have shown how and when some particular
biorthogonal sets, the so-called $\G$-quasi bases, can be used to
define manifestly non self-adjoint operators with known eigenvectors
and simple punctual spectra, even when these eigenvectors do not
form bases for the Hilbert space where the model is defined. In
particular, we have devoted a part of the paper to analyze in some
details the properties of three $\G$-quasi bases, and another part
to show how these sets can be used to define Hamiltonians and ladder
operators, and how the latter  ones can be used to factorize the
Hamiltonians themselves.

Our paper can be seen as another step toward a better comprehension of the role of biorthogonal sets in physical contexts where self-adjointness of the observables is not required. Also, from a more mathematical side, the paper suggests to undertake a deeper analysis of $\G$-quasi bases and of their relations with frames and this is, in fact, one of our future project.

\section*{Acknowledgements}

This work was partially supported by the University of Palermo, by the Gruppo Nazionale per la Fisica Matematica (GNFM) and by the Gruppo Nazionale per l'Analisi Matematica, la
Probabilit\`{a} e le loro Applicazioni (GNAMPA) of the Istituto
Nazionale di Alta Matematica (INdAM).

\end{document}